\documentclass[10pt]{llncs}%
\usepackage{amsfonts}
\usepackage{amsmath}
\usepackage{amssymb}
\usepackage{graphicx}%

\newcommand{\Q}{\mathcal{Q}}

\newcommand{\G}{\mathcal{G}}

\renewcommand{\L}{\mathcal{L}}

\newcommand{\p}{p}
\newcommand{\pt}[1]{p_{#1}}

\newcommand{\comp}[2]{\overline{#1}({#2})}

\renewcommand{\Q}{\mathcal{Q}}

\newcommand{\grq}[2]{G_{{#1},{#2}}}

\newcommand{\indcycle}{\textsf{RCC}}

\begin{document}

\title{On strongly chordal graphs that are not leaf powers}
\author{Manuel Lafond\,$^{1}$}
\institute{$^{1}$ Department of Mathematics and Statistics, University of Ottawa, Ottawa, Canada \\ \email{mlafond2@uOttawa.ca}}

%%%%%%%%%%%%%%%%%%%%%%%%%%%%

\maketitle

\vspace{-7mm}

\begin{abstract}
A common task in phylogenetics is to find an evolutionary tree representing  
proximity relationships between species.
This motivates the notion of leaf powers: 
a graph $G = (V, E)$ is a leaf power if there exist a tree $T$ on leafset $V$ and a threshold $k$  such
that $uv \in E$ if and only if the distance between $u$ and $v$ in $T$ is at most $k$.
Characterizing leaf powers is a challenging open problem, along with determining the complexity of their recognition.
Leaf powers are known to be strongly chordal, but few strongly chordal graphs are known to \emph{not} be leaf powers, as such graphs are difficult to construct.  
Recently, Nevries and Rosenke asked if leaf powers could be characterized 
by strong chordality and a finite set of forbidden induced subgraphs.

In this paper, we provide a negative answer to this question, by exhibiting an infinite family $\G$ of (minimal)
strongly chordal graphs that are not leaf powers.
During the process, we establish a connection between leaf powers, alternating cycles
and quartet compatibility.
We also show that deciding if a chordal graph is $\G$-free is NP-complete.
%, 
%which may provide insight on the complexity of the leaf power recognition problem.
\end{abstract}

\vspace{-8mm}

\section{Introduction}

\vspace{-1mm}

In phylogenetics, a classical method for inferring an evolutionary 
tree of species is to construct the tree from a distance matrix, which depicts how close or far
each species are to one and another.
Roughly speaking, similar species should be closer to  each other in the tree than more distant species.
In some contexts, the actual distances are ignored (e.g. when they cannot be trusted 
due to errors), and only the notions 
of ``close'' and ``distant'' are preserved.  This corresponds to a graph in which the vertices are the species, and two vertices share an edge if and only if they are ``close''.
This motivates the definition of \emph{leaf powers}, which was proposed by Nishimura et al. in~\cite{nishimura2002graph}:
a graph $G = (V, E)$ is a leaf power if there exist a tree $T$ on leafset $V(G)$ and a threshold $k$  such
that $uv \in E$ if and only if the distance between $u$ and $v$ in $T$ is at most $k$.
Hence the tree $T$, which we call a \emph{leaf root}, is a potential evolutionary history for $G$, as it satisfies 
the notions of ``close'' and ``distant'' depicted by $G$.
It is also worth noting that this type of similarity graph is also encountered in the context of gene \emph{orthology} inference, 
which is a special type of relationship between genes (see e.g.~\cite{li2003orthomcl,tatusov2000cog}).
A similarity graph $G$ is used as a basis for the inference procedure, and being able to verify that $G$ is a leaf power would provide a basic test as to whether $G$ correctly depicts similarity, as such graphs are known to contain errors~\cite{lafond2014orthology}.

A considerable amount of work has been done on the topic of leaf powers (see~\cite{calamoneri2016pairwise} for a survey), 
but two important challenges remain open:
to determine the computational complexity of recognizing leaf powers, 
and to characterize the class of leaf powers from a graph theoretic 
point of view.
Despite some interesting results on graph classes that are leaf powers~\cite{brandstadt2008ptolemaic,brandstadt2010rooted,kennedy2006strictly}, 
both problems are made especially difficult 
due to our limited knowledge on graphs that are \emph{not} leaf powers.
Such knowledge is obviously fundamental for the characterization of leaf powers, but 
also important from the algorithmic perspective: 
if recognizing leaf powers is in $P$, a polynomial time algorithm 
is likely to make usage of structures to avoid, 
and if it is NP-hard, a hardness reduction 
will require knowledge of many non-leaf powers in order to generate ``no''
instances.

It has been known for many years that leaf powers must be strongly chordal (i.e. chordal and sun-free).
Brandst{\"a}dt et. al exhibited one strongly chordal non-leaf power by 
establishing an equivalence between leaf powers and NeST graphs~\cite{bibelnieks1993neighborhood,brandstadt2010rooted}.
Recently~\cite{nevries2016towards}, Nevries and Rosenke found seven such graphs, 
all identified by the notion of bad 2-cycles in \emph{clique arrangements},
which are of special use in strongly chordal graphs~\cite{nevries2015characterizing}.
These graphs have at most 12 vertices, and in~\cite{DBLP:journals/corr/NevriesR14}, the authors 
conjecture that they are the only strongly chordal non-leaf powers.
This was also posed as an open problem in~\cite{calamoneri2016pairwise}.
A positive answer to this question would imply a polynomial time algorithm for recognizing leaf powers, as strong chordality can be checked in $O(\min \{m \log n, n^2\})$ time~\cite{paige1987three,spinrad1993doubly}. 
%and checking for a finite set of forbidden (finite) subgraphs is also easy.

In this paper, we unfortunately give a negative answer to this question.
We exhibit an infinite family $\G$ of strongly chordal graphs that are not 
leaf powers, and each graph in this family is minimal for this property
(i.e. removing any vertex makes the graph a leaf power).
This is done by first establishing a new necessary condition for a graph $G$ to be a leaf power, 
based on its \emph{alternating cycles} (which are cyclic orderings of vertices that alternate between an edge and a non-edge).
Namely, there must be a tree $T$ that can satisfy the edges/non-edges of each alternating cycle $C$ of $G$ after (possibly) subdividing some of its edges (see Section~\ref{sec:altcycles} for a precise definition).
This condition has two interesting properties.  First, every graph currently known to not be a leaf power fails to satisfy this condition.  And more importantly, this provides new tools for the construction of novel classes of non-leaf powers.  In particular, alternating cycles on four vertices enforce the leaf root to contain a 
specific \emph{quartet}, a binary tree on four leaves.
This connection lets us borrow from the theory of \emph{quartet compatibility},
which is well-studied in phylogenetics (see e.g.~\cite{bandelt1986reconstructing,berry1999quartet,shutters2013incompatible,steel1992complexity}).
More precisely, we use results from~\cite{shutters2013incompatible} to 
create a family $\G$ of strongly chordal graphs whose $4$-alternating cycles enforce 
a minimal set of incompatible quartets.
We then proceed to show that deciding if a chordal graph $G$ contains a member
of $\G$ as an induced subgraph is NP-complete.
Thus, $\G$-freeness is the first known property of non-leaf powers that we currently ignore how to check in polynomial time.
This result also indicates that if the problem admits a polynomial time algorithm, it will have 
to make use of strong chordality (or some other structural property), 
since chordality alone is not enough to identify forbidden structures quickly.

The paper is organized as follows:
in Section~\ref{sec:prelim}, we provide some basic notions and facts.
In Section~\ref{sec:altcycles}, we establish the connection between 
leaf powers, alternating cycles and quartets, along with its implications.  In Section~\ref{sec:notleafpower}, we exhibit the family $\G$ of strongly chordal graphs that are not leaf powers.  We then show in Section~\ref{sec:hardness}
that deciding if a chordal graph is $\G$-free is NP-complete.

\vspace{-2mm}

\section{Preliminary notions}\label{sec:prelim}

%graph stuff
All graphs in this paper are simple and finite.
For $k \in \mathbb{N}^+$, we use the notation $[k] = \{1, \ldots, k\}$.
We denote the set of vertices of a graph $G$ by $V(G)$, 
its set of edges by $E(G)$, 
and its set of non-edges by $\comp{E}{G}$.  By $G[X]$ we mean the subgraph induced by $X \subseteq V(G)$.
The set of neighbors of $v \in V(G)$ is $N(v)$.
%For a set of graphs $\H$, we say that $G$ is $\H$-free
%if no induced subgraph of $G$ is isomorphic to a graph of $\H$
%(we may write $H$-free when $\H = \{H\}$).
The $P_4$ is the path of length $3$ and the $2K_2$ is the graph 
%on four vertices
%consisting of exactly two edges that do not share a common endpoint.
consisting of two vertex-disjoint edges.
A \emph{$k$-sun}, denoted $S_k$, is the graph obtained by starting from a clique of size $k \geq 3$ 
with vertices $x_1, \ldots, x_k$, then adding vertices 
$a_1, \ldots, a_k$ such that $N(a_i) = \{x_i, x_{i + 1}\}$ for 
each $i \in [k - 1]$ and $N(a_k) = \{x_k, x_1\}$.
A graph is a sun if it is a $k$-sun for some $k$, and $G$ is \emph{sun-free}
if no induced subgraph of $G$ is a sun.

A graph $G$ is \emph{chordal} if it has no induced cycle with four vertices or more, and $G$ is \emph{strongly chordal} if it is chordal 
and sun-free.  A vertex $v$ is \emph{simplicial} if $N(v)$ is a clique, 
and $v$ is \emph{simple} if it is simplicial and, in addition, 
for every $x, y \in N(v)$, one of $N(x) \subseteq N(y) \setminus \{x\}$ or $N(y) \subseteq N(x) \setminus \{y\}$ holds.  An ordering $(x_1, \ldots, x_n)$ of $V(G)$ 
is a \emph{perfect elimination ordering} if, for each $i \in [n]$, 
$x_i$ is simplicial in $G[\{x_i, \ldots, x_n\}]$.
The ordering is \emph{simple} if, for each $i \in [n]$, 
$x_i$ is simple in $G[\{x_i, \ldots, x_n\}]$.
It is well-known that a graph is chordal if and only if 
it admits a perfect elimination ordering~\cite{fulkerson1965incidence}, 
and a graph is strongly chordal if and only if it admits 
a simple elimination ordering~\cite{farber1983characterizations}.

%tree stuff
%In this paper, any tree $T$ is assumed to have its set of leaves 
%in bijection with a set of labels $\L(T)$.
%Hence we may use the terms leaves and labels interchangeably.

Denote by $\L(T)$ the set of leaves of a tree $T$.
We say a graph $G = (V, E)$ is a $k$-leaf power 
if there exists a tree $T$ with $\L(T) = V$ such that 
for any two distinct vertices $u, v \in V$, 
$uv \in E$ if and only if the distance between $u$ and $v$ in 
$T$ is at most $k$.
Such a tree $T$ is called a $k$-leaf root of $G$.
A graph $G$ is a \emph{leaf power} if there exists a positive integer $k$
such that $G$ is a $k$-leaf power.

%For $X \subseteq \L(T)$, 
%the \emph{restriction} of $T$ to $X$, denoted $T|_X$, 
%is the tree obtained by taking the minimal induced subgraph of $T$
%that connects all the leaves in $X$, then successively contracting its degree two vertices until none remains.
A \emph{quartet} is an unrooted binary tree on four leaves (an unrooted tree $T$ is binary if all its internal vertices have degree exactly $3$).
For a set of four elements $X = \{a,b,c,d\}$, 
there exist $3$ possible quartets on leafset $X$ which we denote $ab|cd, ac|bd$ and $ad|bc$, depending on how internal edge separates the leaves.  
We say that $T$ \emph{contains} a quartet $ab|cd$ if 
$\{a,b,c,d\} \subseteq \L(T)$ and 
%$T|_{\{a,b,c,d\}}$ is isomorphic to $ab|cd$ (with preservation of labels).
the path between $a$ and $b$ does not intersect the path between $c$ and $d$.
We denote $\Q(T) = \{ab|cd : T$ contains $ab|cd\}$.
We say that a set of quartets $Q$ 
is \emph{compatible} if there exists a tree $T$ such that
$Q \subseteq \Q(T)$, and otherwise $Q$ is \emph{incompatible}. 
%The minimum $k$ for which $G$ is a $l$-leaf-power is called the \emph{leaf rank}
%of $G$ and is denoted $lr(G)$.

For a tree $T$ and $x, y \in V(T)$, $\pt{T}(x, y)$ denotes the set of edges
on the unique path between $x$ and $y$.  We may write $\p(x, y)$ when $T$ is clear from the context.
It will be convenient to extend the definition of leaf powers to weighted edges.
A \emph{weighted} tree $(T, f)$ is a tree accompanied 
by a function $f: E(T) \rightarrow \mathbb{N}^+$ weighting its edges.
If $F \subseteq E(T)$, we denote $f(F) = \sum_{e \in F}f(e)$.
The distance $d_{T,f}(x, y)$ between two vertices of $T$
is given by $f(\p(x, y))$, i.e. the sum of the weights of the edges lying on the $x-y$ path in $T$.
We may write $d_f(x, y)$ for short.
We say that $(T, f)$ is a \emph{leaf root} of a graph $G$ 
if there exists an integer $k$ such that 
$xy \in E(G)$ iff $d_f(x, y) \leq k$.
We will call $k$ the \emph{threshold} corresponding to $(T, f)$.
Note that in the usual setting, the edges of leaf roots are not weighted,
though arbitrarily many degree $2$ vertices are allowed.
It is easy to see that this distinction is merely conceptual, since 
an edge $e$ with weight $f(e)$ can be made unweighted by subdividing it $f(e) - 1$ times.
%since weights are strictly positive, 
%an unweighted tree preserving the distances of $f$ can be obtained 
%by adding $f(e) - 1$ degree $2$ vertices on each edge $e$ 
%\footnote{And conversely, if one wishes a weighted leaf root without degree $2$ vertices, an unweighted tree can be weighted by 
%contracting its degree $2$ vertices and setting $f(e)$ to be the number 
%of degree $2$ vertices that were contracted on $e$, minus $1$.}.

A tree $T$ is \emph{unweighted} if it is not equipped with a weighting function.  
We say an unweighted tree is an \emph{unweighted leaf root} of a graph $G$
if there is a weighting $f$ of $E(T)$ 
such that $(T, f)$ is a leaf root of $G$.

A first observation that will be of convenience later on is that, even though the usual definition of leaf powers does not allow edges of weight $0$, they do not alter the class of leaf powers.

\begin{lemma}\label{lem:zero-edges}
Let $G$ be a graph, and let $(T, f)$ be a weighted tree in which $\L(T) = V(G)$ and $f(e) \geq 0$ for each $e \in E(T)$.  If there exists an integer $k$ such that 
$uv \in E(G) \Leftrightarrow d_f(u, v) \leq k$, then $T$ is an unweighted leaf root of $G$. 
\end{lemma}

\begin{proof}
If no edge has weight $0$, there is nothing to do.
Otherwise, we devise a weighting function $f'$ for $T$. 
Let $d = \max_{x,y \in V(T)}|\p(x, y)|$.
Set $f'(e) = (d + 1) \cdot f(e)$ for each $e \in E(T)$ having $f(e) > 0$, 
and $f'(e) = 1$ for each $e \in E(T)$ having $f(e) = 0$.
If $d_f(x, y) \leq k$, then $d_{f'}(x, y) \leq (d + 1) k + d$, 
and if $d_f(x, y) \geq k + 1$, then $d_{f'} \geq (d + 1) k + (d + 1)$.
The threshold $(d + 1) k + d$ shows that $T$ is 
an unweighted leaf root of $G$.
\qed
\end{proof}

A tree $T'$ is a \emph{refinement} of a tree $T$ 
if $T$ can be obtained from $T'$ by contraction of edges.
A consequence of the above follows.

\begin{lemma}\label{lem:refinement}
Let $T$ be an unweighted leaf root of a leaf power $G$.
Then any refinement $T'$ of $T$ is also 
an unweighted leaf root of $G$.
\end{lemma}

\begin{proof}
We may take a weighting $f$ such that $(T, f)$ is a leaf root of $G$, refine it in order to obtain $T'$, weight the newly created edges by $0$ and apply Lemma~\ref{lem:zero-edges}.
\qed
\end{proof}

The following was implicitly proved in~\cite{brandstadt2008ptolemaic}. 
We include the proof in the Appendix for the sake of completeness.

\begin{lemma}\label{lem:degone-dominating}
%Suppose that $G$ has a vertex $v$ such that $v$ is either a dominating vertex, 
%or it has degree $1$.
%Then $G$ is a leaf power if and only if $G - v$ is a leaf power.
%if and only if $T - v$ is a u-leaf root of $G - v$.
Suppose that $G$ has a vertex $v$ of degree $1$.
Then $G$ is a leaf power if and only if $G - v$ is a leaf power.
\end{lemma}

\vspace{-6mm}

\section{Alternating cycles and quartets in leaf powers}\label{sec:altcycles}

\vspace{-1mm}

In this section, we restrict our attention to alternating cycles 
in leaf powers, which let us establish a new necessary condition on the 
topology of unweighted leaf roots.  
This will serve as a basis for the construction of our family of forbidden induced subgraphs.
Although we will not use the full generality of the statements proved here, we believe they may be of interest for future studies.
%We will then demonstrate that every graph currently known to not be a leaf power fails to satisfy this condition, 
%mostly through the notion of quartet compatibility.  

%We will be interested in particular subsets 
%of the relations depicted by a graph $G$.
Let $(A, B)$ be a pair such that $A \subseteq E(G)$
and $B \subseteq \comp{E}{G}$. 
We say a weighted tree $(T, f)$ \emph{satisfies} $(A, B)$
if there exists a threshold $k$ such that for each edge
$\{x, y\} \in A$, $d_f(x, y) \leq k$ and 
for each non-edge $\{x, y\} \in B$, $d_f(x, y) > k$.
Thus $(T, f)$ is a leaf root of $G$ iff it satisfies $(E(G), \comp{E}{G})$.
For an unweighted tree $T$, we say that 
$T$  \emph{can satisfy} $(A, B)$ if there exists a weighting $f$
of $E(T)$ such that $(T, f)$ satisfies $(A, B)$.

A sequence of $2c$ distinct vertices $C=(x_0, y_0, x_1, y_1, \ldots, x_{c - 1}, y_{c - 1})$  is an \emph{alternating cycle} of a graph $G$ if
for each $i \in \{0, \ldots, c - 1\}$, $x_{i}y_{i} \in E(G)$ and
$y_{i}x_{i + 1} \notin E(G)$ (indices are modulo $c$ in all notions related to alternating cycles).
In other words, the vertices of $C$ alternate beween an edge and a non-edge.  
We write $V(C) = \{x_0, y_0, \ldots, x_{c - 1}, y_{c - 1}\}$, 
$E(C) = \{x_iy_i : 0 \leq i \leq c - 1\}$ and $\comp{E}{C} = \{y_ix_{i + 1} : 0 \leq i \leq c - 1 \}$. 
A weighted tree \emph{satisfies} $C$ if it satisfies $(E(C), \comp{E}{C})$, 
and an unweighted tree \emph{can satisfy} $C$ if it can satisfy 
$(E(C), \comp{E}{C})$.
The next necessary condition for leaf powers is quite an obvious one, but will be of importance  throughout the paper.

\begin{proposition}\label{prop:must-display-cycles}
If $G$ is a leaf power, then there exists an unweighted tree $T$ 
that can satisfy every alternating cycle of $G$.
\end{proposition}

As it turns out, every graph that is currently known to not be a leaf power 
fails to satisfy the above condition (actually, we may even restrict our attention to cycles of length $4$ and $6$, as we will see).  
This suggests that it is also sufficient, and we conjecture that
if there exists a tree that can satisfy every alternating cycle of $G$, then $G$ is a leaf power.  
As a basic sanity check towards this statement, we show that in the absence of alternating cycles, a graph is indeed a leaf power.

\begin{proposition}\label{prop:altcycle}
If a graph $G$ has no alternating cycle, then $G$ is a leaf power.
\end{proposition}

\begin{proof}
Since a chordless cycle of length at least $4$ contains an alternating 
cycle, $G$ must be chordal.  By the same argument, $G$ cannot contain an induced gem (the gem is obtained by taking a $P_4$, and adding a vertex adjacent to each vertex of the $P_4$).
In~\cite{brandstadt2008ptolemaic}, it is shown that chordal gem-free graphs are leaf powers.
\qed
\end{proof}

We will go a bit more in depth with alternating cycles, by first providing a characterization of the 
unweighted trees that can satisfy an alternating cycle $C$.
Let $T$ be an unweighted tree with $V(C) \subseteq V(T)$.
For each $i \in \{0, \ldots, c - 1\}$, we say the path in $T$ between $x_i$ and $y_i$ 
is \emph{positive}, and the path between 
$y_i$ and $x_{i + 1}$ is \emph{negative} (with respect to $C$).

\begin{lemma}\label{lem:alt-cycle-main}
An unweighted tree $T$ can satisfy an alternating 
cycle \\
$C = (x_0, y_0, \ldots, x_{c - 1}, y_{c - 1})$ if and only 
if there exists an edge $e$ of $T$ that belongs to strictly more 
negative paths than positive paths w.r.t. $C$.
\end{lemma}

\begin{proof}
Due to space constraints, we only prove the ($\Rightarrow$) direction.
The proof of sufficiency is 
relegated to the Appendix.

$(\Rightarrow)$: suppose that no edge is on more negative paths
than positive paths, and yet $T$ can satisfy $C$.
Let $f$ be a weighting such that $(T, f)$ satisfies $C$ with some threshold $k$.
For each $i \in \{0, \ldots, c - 1\}$,  
let $A_i = \p(y_i, x_{i + 1}) \setminus \p(x_{i + 1}, y_{i + 1})$
and $B_i = \p(x_{i + 1}, y_{i + 1}) \setminus \p(y_i, x_{i + 1})$.
Moreover, let $R_i = \p(y_i, x_{i+1}) \cap \p(x_{i + 1}, y_{i + 1})$.
Observe that $f(\p(y_i, x_{i + 1})) = f(A_i) + f(R_i) = f(\p(x_{i + 1}, y_{i+1})) + f(A_i) - f(B_i)$.
%Then $f(\p(y_i, x_{i + 1})) = f(\p(x_{i+1}, y_{i+1})) + f(A_i) - f(B_i)$.
%Given the above, we claim 
We claim that for any integer $j \geq 1$, 

\vspace{-3mm}

$$f(\p(x_0, y_0)) < f(\p(x_j, y_j)) + \sum_{i = 0}^{j - 1}(f(A_i) - f(B_i))$$

\vspace{-1mm}

(where the indices of the $x_j, y_j, A_i$ and $B_i$ are taken modulo $c$).
This is easily proved by induction.
For $j = 1$, 
we have $f(\p(x_0, y_0)) \leq k < f(\p(y_0, x_1)) = f(\p(x_1, y_1)) + f(A_0) - f(B_0)$
since $x_0y_0$ is an edge of $C$ but $y_0x_1$ is not.
For higher values of $j$, 
the same argument can be applied inductively:
suppose $f(\p(x_0, y_0)) < f(\p(x_{j - 1}, y_{j - 1})) + \sum_{i = 0}^{j - 2}(f(A_i) - f(B_i))$.  The claim follows from the fact that
$f(\p(x_{j - 1}, y_{j - 1})) \leq k < f(\p(y_{j - 1}, x_j)) = 
f(\p(y_j, x_{j})) + f(A_{j - 1}) - f(B_{j - 1})$.

Using the above claim, by setting $j = c$, 
we obtain $f(\p(x_0, y_0)) < f(\p(x_0, y_0)) + \sum_{i = 0}^{c - 1}(f(A_i) - f(B_i))$, i.e. $\sum_{i = 0}^{c - 1}f(B_i) < \sum_{i = 0}^{c - 1}f(A_i)$.
%Now, let $R_i = \p(y_i, x_{i+1}) \cap \p(x_{i + 1}, y_{i + 1})$.  
Then $\sum_{i = 0}^{c - 1}(f(B_i) + f(R_i)) < \sum_{i = 0}^{c - 1}(f(A_i) + f(R_i))$.
But since $\p(y_i, x_{i + 1})$ is the disjoint union of $A_i$ and $R_i$, 
and $\p(x_{i + 1}, y_{i + 1})$ the disjoint union of $B_i$ and $R_i$, 
this implies 
$\sum_{i = 0}^{c - 1}f(\p(x_{i + 1}, y_{i + 1})) < \sum_{i = 0}^{c - 1}f(\p(y_i, x_{i + 1}))$.  For any given edge $e$, $f(e)$ is summed as many times as it appears on a positive path on the left-hand side, and as many times as it appears on a negative 
path on the right-hand side.  Since, by assumption, no edge appears on more 
negative than positive paths, we have reached a contradiction since this inequality is impossible.
\qed
\end{proof}

Lemma~\ref{lem:alt-cycle-main} lets us relate quartets and 4-alternating cycles easily.
If $C = (x_0, y_0, x_1, y_1)$, 
the edges of the quartets $x_0x_1|y_0y_1$ and $x_0y_1|y_0x_1$ do not meet the condition
of Lemma~\ref{lem:alt-cycle-main}, and therefore no unweighted leaf root can contain these quartets.
This was already noticed in~\cite{nevries2016towards}, although this was presented in another form and not stated in the language of quartets.

\begin{corollary}\label{cor:qtet}
Let $C = (x_0, y_0, x_1, y_1)$ be a $4$-alternating cycle of a graph $G$. 
Then a tree $T$ can display $C$ if and only if 
$T$ contains the $x_0y_0|x_1y_1$ quartet.
\end{corollary}

We will denote by $RQ'(G)$ the set of \emph{required quartets} of $G$, that is 
$RQ'(G) = \{x_0y_0|x_1y_1 : (x_0, y_0, x_1, y_1)$ is an alternating cycle of $G\}$.
The only graphs on $4$ vertices that contain an alternating cycle are the $P_4$, the $2K_2$
and the $C_4$.  However, the $C_4$ contains two distinct alternating cycles: if four vertices $abcd$ in cyclic order form a $C_4$, then $(a, b, d, c)$ and 
$(d, a, c, b)$ are two alternating cycles.  The first implies the 
$ab|cd$ quartet, whereas the second implies the $ad|cb$ quartet.  This shows that no leaf power can contain a $C_4$.  Thus $RQ'(G)$ can be constructed by enumerating 
the $O(n^4)$ induced $P_4$ and $2K_2$ of $G$.
It is worth mentioning that deciding if a given set of quartets is compatible is NP-complete~\cite{steel1992complexity}.  However, $RQ'(G)$ is not \emph{any} set of quartets since it is generated from $P_4$'s and $2K_2$'s of a strongly chordal graph, and the hardness does not immediately transfer.

Now, denote by $RQ(G)$ the set of quartets that any unweighted leaf root of $G$ must contain, if it exists.
Then $RQ'(G) \subseteq RQ(G)$, and equality does not hold in general. 
Below we show how to find some of the quartets 
from $RQ(G) \setminus RQ'(G)$
(Lemma~\ref{lem:path-lemma}, which is a generalization of~\cite[Lemma 2]{nevries2016towards}).
%First, let $P_1 = x_0x_1 \ldots x_p$ and $P_2 = be a path of $G$ (possibly with chords), and let $y_0y_1 \in E(G)$ with $y_0, y_1$ not belonging to $P$.
%We say that $P$ \emph{alternates on $y_0y_1$} 
%if, for each $i \in \{0, 1, \ldots, p - 1\}$, there is a $b \in \{0, 1\}$
%such that $(x_i, x_{i + 1}, y_b, y_{1-b})$ is an alternating cycle.

\begin{lemma}\label{lem:path-lemma}
Let $P_1 = x_0x_1 \ldots x_p$ and $P_2 = y_0y_1 \ldots y_q$ 
be disjoint paths of $G$ (with possible chords) such that for any $0 \leq i < p$ and $0 \leq j < q$, 
$\{x_i, x_{i + 1}, y_j, y_{j + 1}\}$ are the vertices of an alternating cycle.
Then $x_0x_p|y_0y_q \in RQ(G)$.
%Let $P = x_0x_1 \ldots x_p$ be a path of $G$, and let $y_0y_1$ and $z_0z_1$ be two (not necessarily distinct) edges such that $y_0 \neq z_0$ 
%and $P$ alternates on both $y_0y_1$ and $z_0z_1$
%($y_1 = z_1$, $y_1 = z_0$ and $z_1 = y_0$ are possible).
%Then $x_0x_p | y_0z_0 \in RQ(G)$.
\end{lemma}

\begin{proof}
First note that in general, if a tree $T$ contains
the quartets $ab|c_ic_{i + 1}$ for $0 \leq i < l$, 
then $T$ must contain $ab|c_0c_l$ 
(this is easy to see by trying to construct such a $T$: start with the $ab|c_0c_1$ quartet, and insert $c_2, \ldots, c_l$ in order - at each insertion, $c_i$ cannot have its neighbor on the $a - b$ path).
For any $0 \leq i < p$, we may apply this observation on $\{a,b\} = \{x_i, x_{i + 1}\}$. This yields $x_ix_{i + 1}|y_0y_q \in RQ(G)$, since $x_ix_{i + 1}|y_jy_{j + 1} \in RQ'(G)$ for every $j$.
Since this is true for every $0 \leq i < p$, 
we can apply this observation again, this time on 
$\{a, b\} = \{y_0, y_q\}$ (and the $c_i$'s being the $x_i$'s) and deduce that $y_0y_q|x_0x_p \in RQ(G)$.
%We show by induction that for each $i \in [p]$, 
%$x_0x_{i}|y_0z_0 \in RQ(G)$.
%For $i = 1$, observe that $x_0x_1|y_0y_1 \in RQ'(G)$
%and $x_0x_1|z_0z_1 \in RQ'(G)$, which together 
%enforce $x_0x_1|y_0z_0 \in RQ(G)$.
%For the induction step, suppose $x_0x_i|y_0z_0 \in RQ(G)$
%for $i < p - 1$.
%Note that $x_ix_{i + 1}|y_0y_1, x_ix_{i + 1}|y_0y_1 \in RQ'(G)$.
%Together with $x_0x_i|y_0z_0 \in RQ(G)$, this enforces
%$x_0x_{i+1}|y_0z_0 \in RQ(G)$.
\qed
\end{proof}

In particular, suppose that $G$ has two disjoint pairs of vertices $\{x_0, x_1\}$ and 
$\{y_0, y_1\}$ such that $x_0$ and $x_1$ 
(resp. $y_0$ and $y_1$) 
share a common neighbor $z$ (resp. $z'$), 
and $z \notin N(y_0) \cup N(y_1)$ (resp. $z' \notin N(x_0) \cup N(x_1)$).  Then $x_0x_1|y_0y_1 \in RQ(G)$.
%Also note that Lemma~\ref{lem:path-lemma} implies that 
%$x_{i}x_{j}|y_hy_l \in RQ(G)$ for all $0 \leq i, j \leq p, 0 \leq h, l \leq q$ by taking every possible subpaths of $P_1$ and $P_2$. 

In the rest of this section, we briefly explain how 
all known non-leaf powers fail to satisfy Proposition~\ref{prop:altcycle}.
We have already argued that a leaf power cannot contain a $C_4$.
As for a cycle $C_n$ with $n > 4$ and vertices $x_0, \ldots, x_{n - 1}$
in cyclic order, observe that $x_ix_{i+1}|x_{i+2}x_{i+3} \in RQ(C_n)$ since they form a $P_4$, 
for each $i \in \{0, \ldots, n - 1\}$ (indices are modulo $n$).  In this case it is not difficult to 
show that $RQ(C_n)$ is incompatible, providing an alternative explanation as to why leaf powers must be chordal.

\begin{figure*}[!b]
\begin{center}
\includegraphics[width= 1.0 \textwidth]{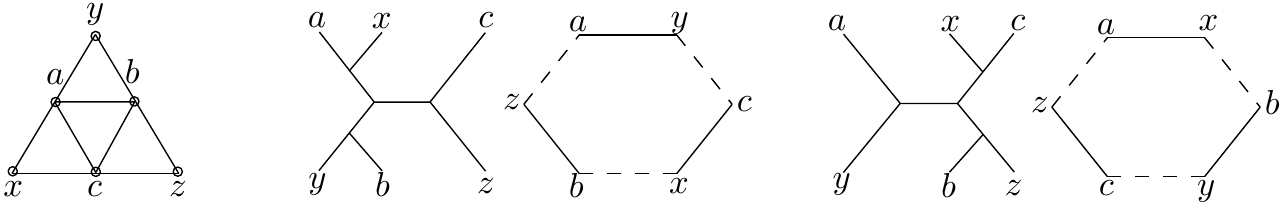}
\caption{The $3$-sun $S_3$, and the two trees that contain $RQ'(S_3) = \{ay|cz, by|cx, bz|ax\}$, 
with each tree accompanied by the alternating cycle of $S_3$ that it cannot satisfy.}\label{fig:three-sun}
\end{center}
\end{figure*}

A similar argument can be used for $S_n$, the $n$-sun, when $n \geq 4$.  If we let $x_0, \ldots, x_{n - 1}$ be the clique vertices of $S_n$ arranged in cyclic order, 
again $x_ix_{i+1}|x_{i+2}x_{i+3} \in RQ(S_n)$ for $i \in \{0, \ldots, n - 1\}$, here because of Lemma~\ref{lem:path-lemma} and the degree $2$ vertices of $S_n$.
Only $S_3$, the $3$-sun, requires an ad-hoc argument, and it is currently the only known non-leaf power
for which the set of required quartets are compatible.
Figure~\ref{fig:three-sun} illustrates how alternating cycles show that $S_3$ is not a leaf power.
There are only two trees that contain $RQ'(S_3) = \{ay|cz, by|cx, bz|ax\}$, and for both, there is an alternating cycle
such that each edge is on the same number of positive and negative paths.
We do not know if there are other examples 
for which quartets 
are not enough to discard the graph as a leaf power.
Moreover, an open question is whether for each even integer $n$, 
there exists a non-leaf power and a tree that can satisfy every alternating cycle of length 
$< n$, but not every alternating cycle of length $n$.

As for the seven strongly chordal graphs presented in~\cite{nevries2016towards},
they were shown to be non-leaf powers by arguing that $RQ(G)$ 
was not compatible (although the proof did not use the language of 
quartet compatibility).

\section{Strongly chordal graphs that are not leaf powers}\label{sec:notleafpower}

We will use a known set of (minimally) incompatible quartets as a basis for constructing our graph family.

\begin{theorem}[\cite{shutters2013incompatible}]\label{thm:incompat-shutters}
For every integers $r, q \geq 3$, the quartets $Q = \{a_ia_{i + 1}|b_jb_{j + 1} : i \in [r - 1], j \in [q - 1]\} \cup \{a_1b_1|a_rb_q\}$ are incompatible.
Moreover, any proper subset of $Q$ is compatible.
\end{theorem}

\begin{figure*}[!t]
\begin{center}
\includegraphics[width= 1.0 \textwidth]{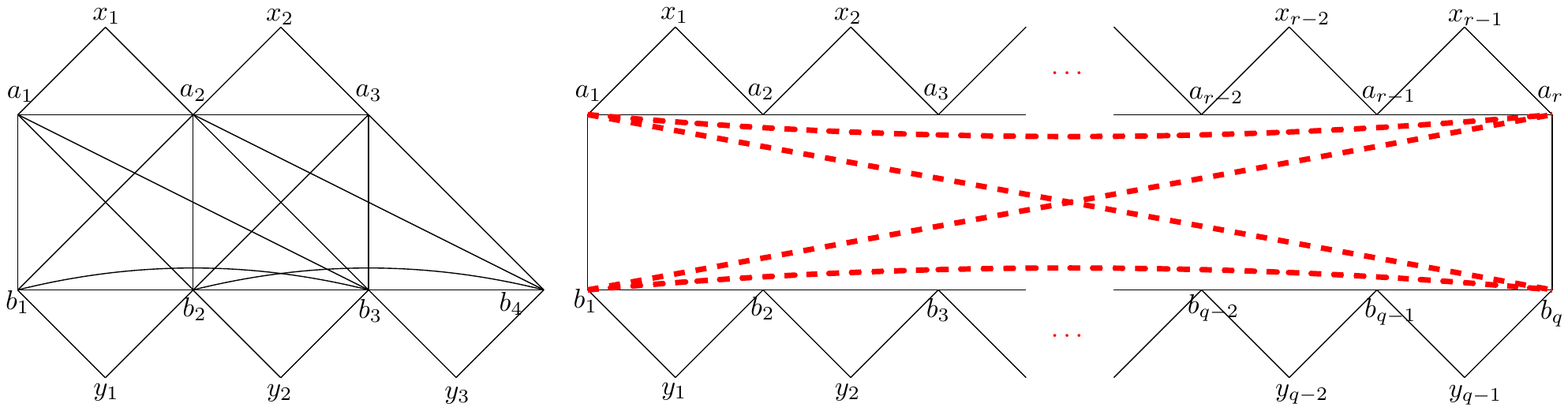}
\caption{The graph $\grq{3}{4}$ on the left, followed by its generalization $\grq{r}{q}$ on the right.  In the latter, all edges between the $a_i$'s and $b_i$'s are present, except the non-edges depicted by red dashed lines.}\label{fig:g_rq}
\end{center}
\end{figure*}

We now construct the family $\{\grq{r}{q} : r, q \geq 3\}$ of minimal strongly chordal graphs
that are not leaf powers.  The idea is to simply enforce that 
$RQ(\grq{r}{q})$ contains all the quartets of $Q$ in  Theorem~\ref{thm:incompat-shutters}.
Figure~\ref{fig:g_rq} illustrate the graph $\grq{3}{4}$ and a general representation of $\grq{r}{q}$.
For integers $r, q \geq 3$, $\grq{r}{q}$ is as follows:
start with a clique of size $r + q$, partition its vertices into two disjoint sets 
$A = \{a_1, \ldots a_r\}$ and 
$B = \{b_1, \ldots, b_q\}$, and remove 
the edges $a_1a_r, a_1b_q, b_1b_q$ and $b_1a_r$.
Then for each $i \in [r - 1]$ insert a node $x_i$
that is a neighbor of $a_i$ and $a_{i + 1}$, 
and for each $i \in [q - 1]$, insert another node $y_i$ that is a neighbor
of $b_i$ and $b_{i + 1}$.

We note that in~\cite{nevries2016towards}, the graph $\grq{3}{3}$ was one of the seven graphs shown to be a strongly chordal non-leaf power.  Hence $\grq{r}{q}$ can be seen as a generalization of this example.
It is possible that the other examples of~\cite{nevries2016towards} can also be generalized.
 
\begin{theorem}
For any integers $r,q \geq 3$, the graph $\grq{r}{q}$ is strongly chordal, is not a leaf power and for any $v \in V(\grq{r}{q})$, $\grq{r}{q} - v$ is a leaf power.
\end{theorem}

\begin{proof}
One can check that $\grq{r}{q}$ is strongly chordal by the simple elimination ordering:
$x_1, x_2, \ldots, x_{r - 1}, y_1, \ldots, y_{q- 1}, a_1, b_1, a_r, b_q,a_2, \ldots, a_{r - 1}, b_2, \ldots, b_{q - 1}$.
%followed by any ordering of the remaining vertices since they form a clique.

To see that $\grq{r}{q}$ is not a leaf power, we note that 
the incompatible set of quartets of Theorem~\ref{thm:incompat-shutters} is a subset of 
$RQ(\grq{r}{q})$:
$a_ia_{i+1}|b_jb_{j+1} \in R(\grq{r}{q})$ by Lemma~\ref{lem:path-lemma} and the paths 
$a_ix_ia_{i + 1}$ and $b_jy_jb_{j+1}$, 
and $a_1b_1|a_rb_q \in RQ(\grq{r}{q})$ since they induce a $2K_2$.

We now show that for any $v \in V(\grq{r}{q})$, $\grq{r}{q} - v$ is a leaf power.  
First suppose that $v \in A \cup B$, say $v = a_i$ without loss of generality.  
Then in $\grq{r}{q} - a_i$, $x_i$ (or take $x_{i - 1}$ if $i = r$) has degree one, and so by Lemma~\ref{lem:degone-dominating}, 
$\grq{r}{q} - a_i$ is a leaf power if and only if $\grq{r}{q} - a_i - x_i$ is a leaf power.
Therefore, it suffices to show that $\grq{r}{q} - x_i$ is a leaf power.
We may thus assume that $v = x_i$ for some $i$ (the $v = y_i$ case is the same by symmetry).

\begin{figure*}[!t]
\begin{center}
\includegraphics[width= 1.0 \textwidth]{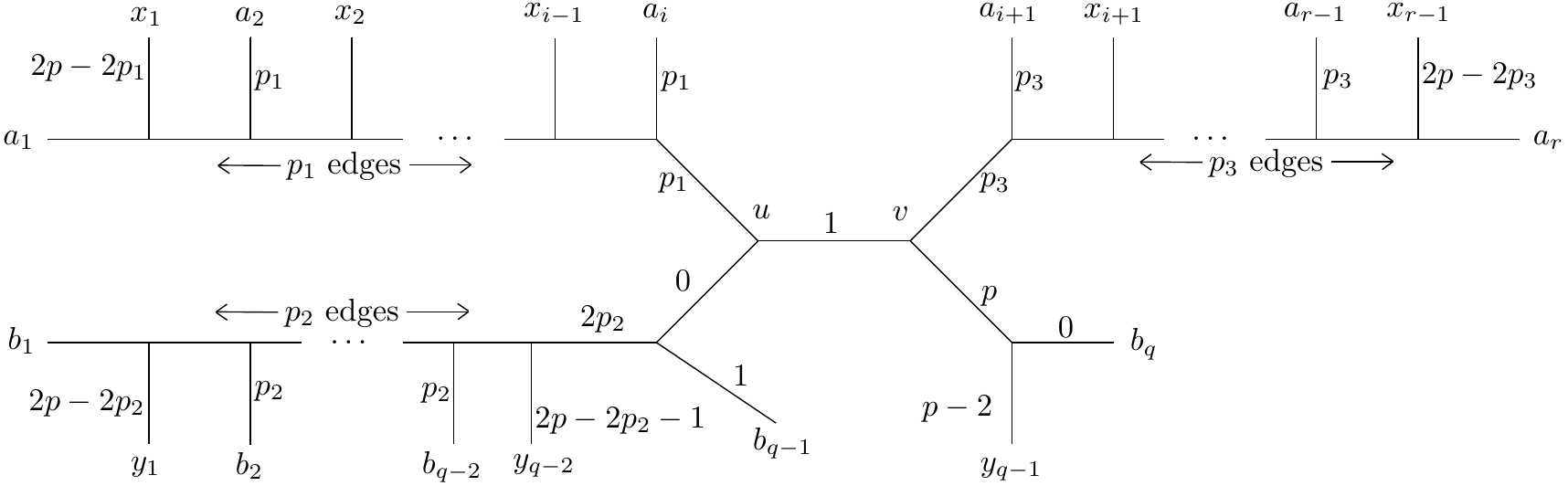}
\caption{A leaf root of $\grq{r}{q} - x_i$.}\label{fig:labeling_rq}
\end{center}
\end{figure*}

Figure~\ref{fig:labeling_rq} exhibits a leaf root $(T, f)$ for $\grq{r}{q} - x_i$ (the weighting contains $0$ edges, but this can be handled by Lemma~\ref{lem:zero-edges}).  %Essentially, the absence of $x_i$ allows us to ``break'' the cyclic quartet structure of $G_t$ and allow the $a_1b_1|a_rb_q$ separation. 
In the weighting $f$, the edges take values depending on variables $p, p_1, p_2, p_3$ which are defined as follows:  
\begin{align*}
p := 2(2i - 1) (2r - 2i - 1) (2q  - 3)  & \>\>\>\>\>\> p_1 := p/(2i - 1)
 \\
p_2 := p/(2q - 3) & \>\>\>\>\>\> p_3 := p/(2r - 2i - 1)
\end{align*}

and we set the threshold $k := 2p$.
Each edge on the $a_1 - u, b_1 - u$ and $a_r - v$ path 
is weighted by $p_1, p_2$ and $p_3$ respectively, with the exception 
of the last two edges of the $b_1 - u$ path where one edge has weight $0$ and the other $2p_2$.
One can check that this ensures that $f(\p(a_1, u)) = f(\p(b_1, u)) = f(\p(a_r, v)) = p$, 
($p_1, p_2$ and $p_3$ are chosen so as to distribute a total weight of $p$ across these paths, and $p$ is such that these values are integers).
Moreover, $p_1,p_2,p_3 > 2$.
Observe that if $i = 1$, then the $a_1 - u$ path is a single edge and $p_1 = p$, and if $i = r - 1$, the $a_r - v$ path is a single edge and $p_3 = p$.
It is not hard to verify that $(T, f)$ satisfies the subgraph of $G - x_i$ induced by the $a_j$'s and $b_j$'s (since each pair of vertices has distance at most $2p$, except $a_1a_r, a_1b_q, b_1a_r$ and $b_1b_q$).

Now for the $x_j$'s and $y_j$'s.
For each $j \in [r - 1] \setminus \{i\}$, the edge $e$ incident to $x_j$ has $f(e) = 2p - 2p_1$ if $j < i$ and $f(e) = 2p - 2p_3$ if $j > i$.  For $j \in [q - 1]$, the edge $e$ incident to $y_j$ 
has $f(e) = 2p - 2p_2$ if $j \leq q - 3$, $f(e) = 2p - 2p_2 - 1$ if $j = q - 2$ and $f(e) = p - 2$ if $j = q - 1$.
Each $x_j$ is easily seen to be satisfied, as the only vertices of $T$ within distance $2p$ of $x_j$ are $a_j$ and $a_{j+1}$.
This is equally easy to see for the $y_j$ vertices, with the exception of $y_{q - 1}$. 
In $(T, f)$, $y_{q - 1}$ can reach $b_q$ and $b_{q - 1}$ within distance $2p$ as required, but we must argue that it cannot reach $a_i$ nor $a_{i + 1}$ (which is enough, since all the other leaves are farther from $y_{q - 1}$.
But this follows from that fact that $p_1, p_3 > 2$.
% Let us first suppose that $2p \geq d_f(y_{q - 1}, a_{i + 1})$.
% If $i \neq r - 1$, there are two $p_3$ edges between $v$ and $a_{i + 1}$, and this implies $2p \geq 2p - p_2 - 1 + 2p_3$ and hence $2p_3 \leq p_2 + 1$.  This would imply 
% $2p/(2t - 2i - 1) \leq p/(2t - 3) + 1$, which can be shown to only be possible 
% when $i \geq t$ for $t \geq 3$ and $p > 5t$, which a contradiction since $x_t$ does not exist.
% If $i = t - 1$, $va_t$ is an edge of weight $p$, and we get $d_f(y_{t - 1}, a_t) = 3p - p_2 - 1 = 3p - p/(2t - 3) - 1 > 2p$.
% The $y_{t - 1}$ versus $a_i$ case can be argued in the same manner.
This shows that $(T, f)$ is a leaf root of $\grq{r}{q} - x_i$, and concludes the proof.
\qed
\end{proof}

Interestingly, the $\grq{r}{q}$ graphs might be subject to various alterations in order to
obtain different families of strongly chordal non-leaf powers.
One example of such an alteration of $\grq{r}{q}$ is to pick 
some $j \in \{2, \ldots, r - 2\}$
and remove the edges $\{a_ib_q : 2 \leq i \leq j\}\}$.
One can verify that the resulting graph is still strongly chordal, but
requires the same set of incompatible quartets as $\grq{r}{q}$.

\vspace{-3mm}

\section{Hardness of finding $\grq{r}{q}$ in chordal graphs}\label{sec:hardness}

\vspace{-2mm}

We show that deciding if a chordal graph contains an induced subgraph isomorphic to $\grq{r}{q}$ for some $r, q \geq 3$ is NP-complete.
We reduce from the following:

\vspace{2mm}

\noindent The \textsf{Restricted Chordless Cycle} (\indcycle) problem: \\
\noindent {\bf Input}: a bipartite graph $G = (U \mathbin{\mathaccent\cdot\cup} V, E)$, and two vertices $s, t \in V(G)$ 
such that $s, t \in U$, both $s$ and $t$ are of degree $2$ and they share no common neighbor. \\
\noindent {\bf Question}:  does there exist a chordless cycle in $G$ containing both $s$ and $t$?\\

\vspace{-2mm}

The \indcycle~problem is shown to be NP-hard in~\cite[Theorem 2.2]{diot2014detecting}\footnote{Strictly speaking, the problem asks if there exists a chordless cycle with both $s$ and $t$ \emph{of size at least $k$}.  However, 
in the graph constructed for the reduction, any chordless cycle containing $s$ and $t$ has size at least $k$ if it exists - therefore the question of existence is hard.  Also, $s$ and $t$ are not required to be in the same part of the bipartition, but again, 
this is allowable by subdividing an edge incident to $s$ or $t$.}.
We first need some notation.
If $P$ is a path between vertices $u$ and $v$, 
we call $u$ and $v$ its \emph{endpoints}, and 
the other vertices are \emph{internal}.
Two paths $P_1$ and $P_2$ of a graph $G$ are said \emph{independent} if 
$P_1$ and $P_2$ are chordless, 
do not share any vertex except perhaps their endpoints, and for any internal vertices $x$ in $P_1$ and $y$ in $P_2$, $xy \notin E(G)$.
Observe that there is a chordless cycle containing $s$ and $t$
if and only if there exist two independent paths $P_1$ and $P_2$ 
between $s$ and $t$.

From a \indcycle~instance $(G,s,t)$ 
we construct a graph $H$ for the problem of deciding if $H$ contains an induced copy of $\grq{r}{q}$.
Figure~\ref{fig:reduction} illustrates the construction.

Let $V(H) = \{s_1, t_1, s_2, t_2\} \cup X_U \cup X_V$, 
where $X_U = \{u' : u \in U \setminus \{s, t\}\}$ and 
$X_V = \{v' : v \in V\}$.
%We may denote by $h(u)$ the vertex corresponding to $u$ in $H$
%(i.e. $h(u) = u'$).
Denote $X_U^* = X_U \cup \{s_1, t_1, s_2, t_2\}$. 
For $E(H)$, add an edge between every two vertices of $X_U^*$
\emph{except} the edges $s_1t_1, s_1t_2, s_2t_1,s_2t_2$.
Moreover, we add an edge between $u' \in X_U$ and $v' \in X_V$  if and only if $uv \in E(G)$. Let $\{c_1,c_2\} = N(s)$ and $\{d_1, d_2\} = N(t)$. 
%Let $c_1, c_2$ (resp. $d_1, d_2$) be the two neighbors 
%of $s$ (resp. of $t$).  
Then add edges $s_1c_1'$ and $s_2c_2'$, 
and add the edges $t_1d_1'$ and $t_2d_2'$. 
Notice that $X_V$ forms an independent set.

\begin{figure*}[!t]
\begin{center}
\includegraphics[width= 0.8 \textwidth]{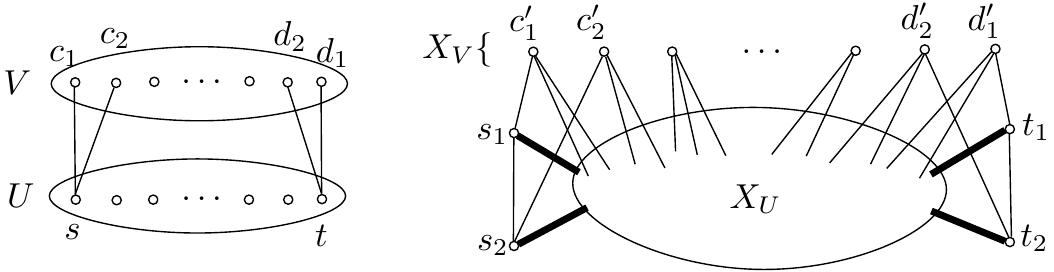}
\caption{An illustration of the reduction: $G$ is on the left (only edges incident to $s$ and $t$ are drawn), $H$ is on the right (thick edges mean that every possible edge is present).}\label{fig:reduction}
\end{center}
\end{figure*}

We claim that  $H$ is chordal.
Note that each vertex $v$ of $X_V$ is simplicial,
since $N(v)$ consists of vertices from $X_U$ and at most
one of $\{s_1, t_1, s_2, t_2\}$ (since $s$ and $t$ have no common neighbor).  Moreover, $H - X_V$ 
is easily seen to be chordal, and it follows that $H$ admits a perfect 
elimination ordering.  %This construction yields the following.

\begin{theorem}
Deciding if a graph $H$ contains a  copy of $\grq{r}{q}$ for some $r, q \geq 3$ is NP-complete, even if $H$ is restricted to the class of chordal graphs.
\end{theorem}

\begin{proof}
The problem is in NP, since a subset $I \subseteq V(H)$, along with the labeling of $I$ by the $a_i$, $b_i$, $x_i$ and $y_i$'s of a $\grq{r}{q}$ can serve as a certificate.
As for hardness, let $G$ be a graph and $H$ the corresponding graph constructed as above.
We claim that $G$ contains two independent paths $P_1$ and $P_2$ between $s$ and $t$ 
if and only if $H$ contains a copy of $\grq{r}{q}$ for some $r, q \geq 3$.
The idea is that $s_1, s_2$ (resp. $t_1, t_2$) correspond to the $a_1, b_1$  (resp. $a_r, b_q$) vertices of $\grq{r}{q}$, while the $P_1$ and $P_2$ paths give the other vertices.  The $x_i$ and $y_i$'s are in $X_V$, and the $a_i$ and $b_i$'s in $X_U^*$.

($\Rightarrow$)  Let $P_1$ and $P_2$ be two independent 
paths between $s$ and $t$.  Note that both paths alternate between $U$ and $V$,
Let $P_1 = (s = u_1,v_1, u_2, v_2, \ldots, v_{r - 1}, u_r = t)$ and $P_2 = (s = w_1,z_1, w_2, z_2, \ldots, z_{q - 1}, w_q = t)$. 
Note that since $P_1$ and $P_2$ are independent, 
every vertex of $G[V(P_1) \cup V(P_2)]$ has degree exactly $2$.

We show that the set of vertices 
%$I = \{s_1,t_1,s_2,t_2\} \cup \{u_i', v_i' \in X_U \cup X_V : i \in [r - 1]\} \cup \{w_i', z_i' \in X_U \cup X_V : i \in [q - 1]\}$ 
$I = \{s_1,t_1,s_2,t_2\} \cup \{x' : x \in V(P_1) \cup V(P_2) \setminus \{s, t\}\}$
forms a $\grq{r}{q}$.
Denote $I_U = I \cap X_U$ and $I_V = I \cap X_V$.
First observe that $\{s_1,t_1,s_2,t_2\} \cup I_U$ forms a clique, but minus the edges $\{s_1t_1, s_1t_2, s_2t_1, s_2t_2\}$.
Hence $\{s_1, s_2\}$ will correspond to the vertices $\{a_1, b_1\}$ of $\grq{r}{q}$, and $\{t_1, t_2\}$ to $\{a_r, b_q\}$, and it remains to find the degree two vertices around this ``almost-clique''.
Observe that $\{v_1, z_1\} = \{c_1, c_2\}$ and $\{v_{r - 1}, z_{q - 1}\} = \{d_1, d_2\}$.
Let $c_{i_1} = v_1, c_{i_2} = z_1$ and $d_{j_1} = v_{r - 1}, d_{j_2} = z_{q - 1}$, 
with $\{i_1, i_2\} = \{j_1, j_2\} = \{1, 2\}$.
In $H$, the vertex sequence $(s_{i_1}, u_2', \ldots, u_{r - 1}', t_{j_1})$ forms a path in $G[I]$ in which every two consecutive vertices share a common neighbor, which lies in $I_V$.
Namely, $s_{i_1}$ and $u_2'$ share $v_1' = c_{i_1}'$, $u_i',u_{i + 1}'$ share $v_i'$, and $u_{r - 1}', t_{j_1}$ share $v_{r - 1}' = d_{j_1}'$.
The same property holds for the consecutive vertices of the path
$(s_{i_2}, w_2', \ldots, w_{q - 1}', t_{i_2})$.
Note that these two paths are disjoint in $H$ and partition $I_U$.
Moreover, by construction each $x' \in I_V$ is a shared 
vertex for some pair of consecutive vertices, i.e. $x'$ has at least two neighbors in $I$. 

Therefore, it only remains to show that if $x' \in I_V$, then $x'$ has only two neighbors in $I$.
Suppose instead that $x'$ has at least $3$ neighbors in $I$, say $y_1', y_2', y_3'$.  Note that all three lie in $X_U^*$.
We must have $|\{s_1, s_2\} \cap \{y_1', y_2', y_3'\}| \leq 1$, 
since $s_1$ and $s_2$ share no neighbor in $X_V$.
Likewise, $|\{t_1, t_2\} \cap \{y_1', y_2', y_3'\}| \leq 1$.
This implies that $y_1', y_2', y_3'$ are vertices corresponding to three distinct vertices of $G$, say $y_1, y_2$ and $y_3$.
Then $x$ is a neighbor of $y_1, y_2, y_3$ and 
since, by construction, $x, y_1, y_2, y_3 \in V(P_1) \cup V(P_2)$, 
this contradicts that $G[V(P_1) \cup V(P_2)]$ has maximum degree $2$.

($\Leftarrow$) Suppose there is $I \subseteq V(H)$ such that $H[I]$ is isomorphic to $\grq{r}{q}$ for some $r, q \geq 3$.   
Add a label to the vertices of $I$ as in Figure~\ref{fig:g_rq} 
(i.e. we assume that we know where the $a_i$'s, $b_i$'s, $x_i$'s and $y_i$'s are in $I$).
We first show that $a_1, b_1, a_r, b_q$, which we will call the \emph{corner} vertices,  
are $s_1, s_2, t_1, t_2$.
%Let $x,y \in I$ be two adjacent corner vertices 
%of the $\grq{r}{q}$ copy induced by $I$, and let
%$w, z$ be the two other adjacent corner vertices.
If one of $a_1$ or $b_1$ is in $X_U$, then 
both $a_r$ and $b_q$ must be in $X_V$, as otherwise there would be an edge between 
$\{a_1, b_1\}$ and $\{a_r, b_q\}$.
%none of $a_r, b_q$ can be in $X_U^*$ since no edge is shared 
%between $\{a_1, b_1\}$ and $\{a_r, b_q\}$.
But $a_r$ and $b_q$ must share an edge, whereas $X_V$ is an independent set.
Thus we may assume 
$\{a_1, b_1\} \cap X_U = \emptyset$.
Suppose that $a_1$ or $b_1$ is in $X_V$, say $a_1$.
Because $b_1 \notin X_U$ as argued above, we must have $b_1 \in \{s_1, s_2, t_1, t_2\}$.
Suppose w.l.o.g. that $b_1 = s_1$.
Hence $a_1 = c_1'$.
Now consider the location of the $x_1$ vertex of $\grq{r}{q}$.
Then $x_1$ must be in $X_U$, in which case $x_1$ is a neighbor of $s_1 = b_1$, 
contradicting that $I$ is a copy of $\grq{r}{q}$.
Therefore, we may assume that $\{a_1, b_1\} \cap X_V = \emptyset$.  By applying the same argument on $a_r$ and $b_q$, we deduce that $\{a_1, b_1, a_r, b_q\} = \{s_1, s_2, t_1, t_2\}$.  
We will suppose, without loss of generality, 
that $a_1 = s_1, b_1 = s_2$ and $\{a_r, b_q\} = \{t_1, t_2\}$
(otherwise we may relabel the vertices of the $\grq{r}{q}$ copy, though note that in doing so we cannot make assumptions on which $t_i$ corresponds to which of $\{a_r, b_q\}$). 

Now let $(s_1 = a_1, a_2, \ldots, a_r = t_j)$, $j \in \{1, 2\}$ be the 
path between the ``top'' corners of the $\grq{r}{q}$ copy in $H$, 
such that $a_ia_{i + 1}$ share a common neighbor $x_i$ of degree $2$ in $G[I]$, $i \in [r - 1]$.
Similarly, let $(s_2 = b_1, b_2, \ldots, b_q = t_l)$,
($l \in \{1, 2\}$ and $l \neq j$) be
the path between the ``bottom'' corners of $\grq{r}{q}$, such that $b_ib_{i + 1}$ share a common neighbor $y_i$ of degree $2$ for $i \in [q - 1]$.
We claim that $a_i \in X_U$ for each $2 \leq i \leq r - 1$.  Suppose instead that some $a_i$ is not in $X_U$.  Since $s_1 = a_1$ is a neighbor of $a_i$, we must have $a_i = c_1'$ (the only other possibility is $a_i = s_2$, but $s_2 = b_1$).  The common neighbor $x_{i - 1}$ of $a_{i - 1}$ and $a_i$  therefore lies in $X_U$.  But then, $x_{i - 1}$ is a neighbor of $s_2 = b_1$, which is not possible.  Therefore, each $a_i$ belongs to $X_U$.  By symmetry, each $b_i$ also belongs to $X_U$.
This implies that every $x_i$ and $y_i$ belong to $X_V$, with $x_1 = c_1', x_{r - 1} = d_j', 
y_1 = c_2'$ and $y_{q - 1} = d_l'$.

We can finally find our independent paths $P_1$ and $P_2$.  
It is straightforward to check that $\{s, t, c_1\} \cup \{u : u' \in \{x_i, a_i\} $ for  $2 \leq i \leq r - 1\}$ induces a path $P_1$ 
from $s$ to $t$ in $G$.
Similarly, 
$\{s, t, c_2\} \cup \{u : u' \in \{y_i, b_i\}$ for  $2 \leq i \leq q - 1\}$ also induces a path $P_2$ from $s$ to $t$.
Moreover, $P_1$ and $P_2$ share no internal vertex.
%To see this, observe that $x_1 = h_{s_1}$ and 
%$x_2 = h_v$, for some $v$, share a common neighbor
%$a_2$.  This implies that $a_2 = \{s, v\}$ and that 
%$sv \in E(G)$.  
%Going forward, for $2 \leq i \leq r - 2$, we have $x_i = h_v$ and $x_{i + 1} = h_w$ sharing a common neighbor $a_{i + 1} = \{v, w\}$, implying that $vw \in E(G)$.  Similarly, if $x_{r - 1} = h_v$, 
%we must have $vt \in E(G)$.  Thus, by following the 
%$x_i$'s in increasing order of $i$, we form a path
%from $s$ to $t$.
%The exact same reasoning can be used to show that 
%$\{s, t\} \cup \{v : h_v = y_i$ for $2 \leq i \leq q - 2\}$ also induces a path $P_2$ from $s$ to $t$.

It only remains to show that $P_1$ and $P_2$ are independent, i.e. form an induced cycle.  
We prove that $G[V(P_1) \cup V(P_2)]$ has maximum degree $2$.
Suppose there is a vertex $v$ of degree at least $3$ in $G[V(P_1) \cup V(P_2)]$.
Then $v \notin \{s, t\}$ since they have degree $2$ in $G$.  
Moreover, $v \notin V$, as otherwise, $v' \in I_V$ 
which implies that $v'$ is an $x_i$ or a $y_i$ and, by construction, $v'$ has at least $3$ neighbors in $I$, a contradiction.
Thus $v \in U$, and its $3$ neighbors lie in $V$.
Hence, $v'$ is either an $a_i$ or a $b_i$ and has three neighbors in $I_V$, which is again a contradiction.
This concludes the proof.
\qed
\end{proof}

\vspace{-6mm}

\section{Conclusion}

\vspace{-1mm}

In this paper, we have shown that leaf powers cannot be characterized by strong chordality and a finite set of forbidden subgraphs.
However, many questions asked here may provide more insight on 
leaf powers.  For one, is the condition of Proposition~\ref{prop:must-display-cycles}
sufficient?  And if so, can it be exploited for some algorithmic or graph theoretic purpose?  Also, we do not know if large alternating cycles are important, 
since so far, every non-leaf power could be explained by checking its alternating cycles of length $4$ or $6$.  A constant bound on the length of ``important'' alternating cycles
would allow enumerating them in polynomial time.

Also, we have exhibited an infinite family of strongly chordal non-leaf powers (along with some variations of it), but it is likely that there are others.  
One potential direction is to try to generalize all of the seven graphs found in~\cite{nevries2016towards}.   The clique arrangement
of $\grq{r}{q}$ may be informative towards this goal.
Finally on the hardness of recognizing leaf powers, the hardness of finding $\grq{r}{q}$
in strongly chordal graphs is of special interest.  
A NP-hardness proof would now be significant evidence towards the difficulty of deciding leaf power membership.
And in the other direction, a polynomial time recognition algorithm may provide important insight on how 
to find forbidden structures in leaf powers.

\bibliographystyle{plain}
\bibliography{main}

\newpage

\section*{Appendix}

\subsection{Proof of Lemma~\ref{lem:degone-dominating}}

\begin{proof}
For the non-obvious direction, suppose that $(T, f)$ is a leaf root of $G - v$ with threshold $k$.
%If $v$ is a dominating vertex, let $h = \max_{x, y \in V(T)} d_{f}(x, y)$.  Obtain $T'$ by adding $v$ as a neighbor of any internal node $v'$ in $T$.  Get a new weighting $f'$ from $f$ 
%by letting $f'(e) = f(e) + h$ for every edge incident to a leaf other than $v$ in $T$, and let $f'(vv') = 1$.  Let the new threshold be $k' := k + 2h$.
%Observe that for any  leaves $x, y$ of $T' - v$, $d_{f'}(x, y) = d_{f}(x, y) + 2h$, and so their relative distance w.r.t $k'$ hasn't changed.
%Let $f'(vv') = 1$, where $v'$ is the neighbor of $v$ in $T$.  We have that for any $x \in V(T), 
%f'(v, x) = f'(v', x) + 1 \leq f(v', x) + h + 1 \leq 2h + 1 \leq k'$,as desired.
%If $v$ has a single neighbor $w$ in $G$, 
Let $w$ be the neighbor of $v$ in $G$.  Then in $(T, f)$ we can do 
the following modification: 
%first multiply $f(e)$ by $2$ for every edge $e$, multiply $k$ 
%by $2$ (notice that $(T, f)$ is still a leaf root of $G$).
subdivide the edge $wz$ incident to $w$ into two edges $wz'$ and 
$z'z$, set $f(wz') = 0, f(z'z) = f(wz)$, and insert $v$ by making it 
adjacent to $z'$ and setting $f(vz') = k$.  As $f(z'z) > 0$, 
$(T, f)$ satisfies the distance constraints on $v$, and we can apply 
Lemma~\ref{lem:zero-edges}
for the $0$ edge.
\qed
\end{proof}

\subsection*{Proof of Lemma~\ref{lem:alt-cycle-main}, sufficiency}

\begin{figure*}[!t]
\begin{center}
\includegraphics[width= 1.0 \textwidth]{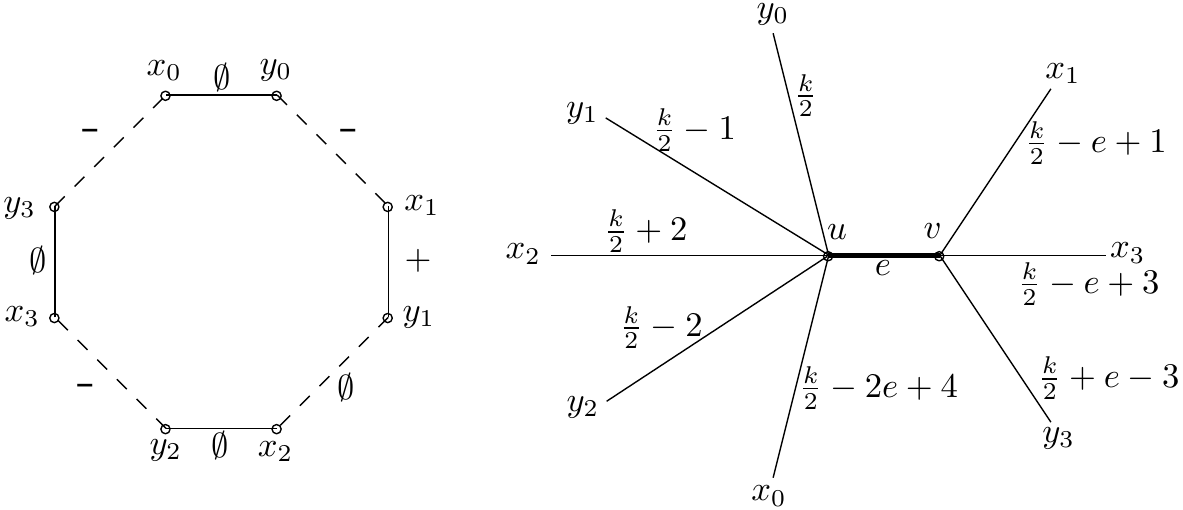}
\caption{An example alternating cycle $C$ and tree $T$ with a single internal edge $uv$.  Edges and non-edges corresponding to a positive and negative path containing $uv$ are labeled with a ``$+$'' and ``$-$'', respectively, while the ``$\emptyset$'' label is for edges/non-edges corresponding to a path that is not negative nor positive.
Here $y_0$, $y_2$ and $y_3$ are negative, whereas $y_1$ is positive.
The tree is labeled with the weights that would be given by the greedy procedure of the proof, starting at $y_0$.}\label{fig:altcycle}
\end{center}
\end{figure*}

$(\Leftarrow)$: let $uv$ be an edge that appears on strictly more negative paths than positive paths.  
We prove that if $u$ and $v$ are the only two internal vertices of $T$ (i.e. every leaf is adjacent to either $u$ or $v$, as in Figure~\ref{fig:altcycle}), then the statement holds.  By Lemma~\ref{lem:refinement}, this is sufficient, since we can refine $T$ into the desired tree afterwards.

Some more notation is required.
For $z \in V(C)$, denote by $a(z)$ the single neighbor of 
$z$ in $T$ (with $a(z) \in \{u, v\}$).
For a weighting function $f$, we may write $f(z)$ instead of $f(z, a(z))$.
Let $X_0 = N(u) \setminus \{v\}$ and $X_1 = N(v) \setminus \{u\}$.
Call two vertices $x, y$ \emph{separated} if $x \in X_i$
and $y \in X_{1 - i}$ for some $i \in \{0,1\}$.
Call a vertex $y_i \in V(C)$ \emph{positive} if $y_i$ and 
$x_{i - 1}$ are not separated, but $y_i$ and $x_i$ are (i.e. $y_i$ gives a ``positive charge'' to $uv$).
Likewise, $y_i$ is \emph{negative} if $y_i$ and $x_{i - 1}$ are separated but $y_i$ and $x_i$ are not.
Otherwise, $y_i$ is \emph{neutral}.
%We shall say that $x_iy_i \in E(C)$ is \emph{positive} 
%if $\p(x_i, y_i)$ goes through $uv$, and \emph{neutral} otherwise.
%Similarly, $y_ix_{i+1} \in \comp{E}{C}$ is \emph{negative}
%if $\p(y_i, x_{i+1})$ goes through $uv$, and \emph{neutral} otherwise.
%A vertex $y_i \in V(C)$ is \emph{positive} if $x_iy_i$ is positive and $y_ix_{i + 1}$ is neutral, \emph{negative} if $x_iy_i$ is neutral and $y_ix_{i + 1}$ is negative, 
%and $y_i$ is \emph{neutral} otherwise.

Let $P, N$ and $Z$ denote, respectively, the 
number of positive, negative and neutral vertices among the $y_i$'s.
Then we must have $N \geq P + 2$.
To see this, first note that the number of paths, positive or negative, that go through $uv$ must be even (since for each path that starts in $X_0$ and goes to $X_1$, there must be a corresponding path from $X_1$ to $X_0$).  Thus $uv$ is on at least two more negative paths than positive paths.  
Then $N \geq P + 2$ follows, since negative vertices correspond to a negative path that cannot be matched with a positive path. %.Having $N < P + 2$ would contradict this property
%since negative vertices correspond to a negative path that cannot be matched with a positive path.

We now construct a weighting $f$ of $T$ so that it satisfies $C$.
Put $e := f(uv) = c^2$ and set $k := 2c^{10}$ to be the threshold for $T$.\footnote{We multiply $k$ by $2$ to ensure it is even}
Suppose that $y_0$ is negative (otherwise relabel vertices), and set $f(y_0) = k/2$.  Then traverse $C$ in cyclic order starting from $x_1$ towards $y_1$ until $x_0$ is reached, weighting each vertex $z$ encountered in a greedy manner, as follows:

\begin{itemize}

\item
if $z = x_i$ and $x_i$ is separated from $y_{i - 1}$, 
set $f(x_i) = k + 1 - e - f(y_{i - 1})$;

\item
if $z = x_i$ and $x_i$ is not separated from $y_{i - 1}$, set $f(x_i) = k + 1 - f(y_{i - 1})$;

\item
if $z = y_i$ and $y_i$ is separated from $x_i$, set $f(y_i) = k - e - f(x_i)$;

\item
if $z = y_i$ and $y_i$ is not separated from $x_i$, set $f(y_i) = k - f(x_i)$.

\end{itemize}

%Note that not every edge is weighted by $f$ during this process, but every edge on the $(a(x_i), y_{i - 1})$ path or the 
%$(a(y_i), x_i)$ path has been weighted in the previous steps.
At the end of this process, every edge of $T$ will be weighted.  
Refer to Figure~\ref{fig:altcycle} for an example application of this procedure.
One can check that all edge weights are positive integers since $f(z) \geq k/2 - ce$ for all $z \in V(C)$.
When the process stops at $x_0$, by construction $f$ satisfies every edge and non-edge of $C$, except possibly the $x_0y_0$ edge of $C$.
Hence it suffices to show that the above weighting satisfies 
$d_f(x_0, y_0) \leq k$.
Since $y_0$ is negative, it is not separated from $x_0$, and hence we must show that $f(x_0) \leq k/2$.

We have $f(x_1) = k/2 + 1 - e$.
Now for $i \in \{1, \ldots, c - 1\}$, consider 
$\Delta_i := f(x_{i + 1}) - f(x_i)$ (recall that indices are modulo $c$).
If $y_i$ is positive, then $f(y_i) = k - f(x_i) - e$
and $f(x_{i + 1}) = k + 1 - f(y_i) = f(x_i)  + e + 1$,
implying $\Delta_i = e + 1$.
Using the same logic on the other cases, we obtain that if 
$y_i$ is negative, $\Delta_i = -e + 1$ and 
if $y_i$ is neutral, $\Delta_i = 1$.
Since $f(x_0) = f(x_1) + \sum_{i = 1}^{c - 1} \Delta_i$, 
we have $f(x_0) = f(x_1) + P \cdot (e + 1) + (N - 1) \cdot (-e + 1) + Z$ (we must use $N - 1$ instead of $N$ because $y_0$ is negative, but is not between $x_1$ and $x_0$ in the visited cyclic order).
This yields $f(x_0) \leq k/2 + 1 - e + e(P - N + 1) + c$. Since $N \geq P + 2$ and $e = c^2$, we obtain $f(x_0) \leq k/2$ as desired.

\end{document}